\documentclass{llncs}
\usepackage{llncsdoc}

\newcommand{\commentOut}[1]{}
\usepackage{amssymb, amsmath, amsfonts}
\usepackage{algorithm}
\usepackage{algpseudocode}

\newcommand{\dividerline}{\begin{center}\hrule\end{center}}

\newcommand{\reduce}{\texttt{reduce}}

\newcommand{\healthy}{\texttt{healthy}}
\newcommand{\cured}{\texttt{cured}}
\newcommand{\faulty}{\texttt{faulty}}

\newcommand{\M}{\texttt{max\_state}}
\newcommand{\m}{\texttt{min\_state}}

\newcommand{\num}{\texttt{nTrim}}

\newcommand{\fillbox}{\hspace*{\fill}\(\Box\)}

\begin{document}
\markboth{\LaTeXe{} Class for Lecture Notes in Computer
Science}{\LaTeXe{} Class for Lecture Notes in Computer Science}
\thispagestyle{empty}

\title{An Improved Approximate Consensus Algorithm \\in the Presence of Mobile Faults}

\author{\normalsize Lewis Tseng}
\institute{\normalsize Computer Science\\
	\normalsize Boston College\\~\\ 
	\normalsize Email: lewis.tseng@bc.edu}

\date{}
\maketitle

~

\begin{abstract}
	\normalsize
	This paper explores the problem of reaching approximate consensus in synchronous point-to-point networks, where each pair of nodes is able to communicate with each other directly and reliably. We consider the \textit{mobile Byzantine fault model} proposed by Garay '94 -- in the model, an omniscient adversary can corrupt up to $f$ nodes in each round, and at the beginning of each round, faults may ``move'' in the system (i.e., different sets of nodes may become faulty in different rounds). Recent work by Bonomi et al. '16 proposed a simple iterative approximate consensus algorithm which requires at least $4f+1$ nodes. This paper proposes a novel technique of using ``confession'' (a mechanism to allow others to ignore past behavior) and a variant of reliable broadcast to improve the fault-tolerance level. In particular, we present an approximate consensus algorithm that requires only $\lceil 7f/2\rceil + 1$ nodes, an $\lfloor f/2 \rfloor$ improvement over the state-of-the-art algorithms. Moreover, we also show that the proposed algorithm is \textit{optimal} within a family of \textit{round-based algorithms}.
\end{abstract}

\newpage

\setcounter{page}{1}

\section{Introduction}

Fault-tolerant {\em consensus} has received significant attentions over the past three decades since the seminal work by Lamport, Shostak, and Pease \cite{lamport_agreement2}. 
Recently, a new type of fault model -- \textit{mobile fault model} -- has been proposed to address  the needs in emerging areas such as mobile robot systems, sensor networks, and smart phones \cite{mobile_yung}. The mobile fault model (in the round-based computation systems) has the following two characteristics:

\begin{itemize}
	\item Up to $f$ nodes may become faulty in a given round, and
	\item Different sets of nodes may become faulty in different rounds.
\end{itemize}
This type of fault model is very different from the traditional ``fixed'' fault model \cite{lamport_agreement2,AA_nancy,welch_book} -- once a node becomes faulty, it remains faulty throughout the lifetime of the computation. 

The mobile fault model is motivated by the observation that for long-living computations, e.g., aggregation, leader election, and clock synchronization, nodes may experience different phases throughout the lifetime such as cured/curing, healthy, and faulty phases \cite{mobile_yung}. For example, a worm-type of malware may gradually infect and corrupt healthy nodes while some infected nodes detected the malware and became cured (e.g., by routine checks from administrators) \cite{mobile_yung}. Another example is that fragile sensor nodes or robots may be impacted by the environment change, e.g., sensor malfunction due to high wind \cite{mobile_bonomi}. 

A rich set of mobile Byzantine fault models has been proposed \cite{mobile_Bonnet2014,mobile_buhrman,mobile_garay,mobile_Sasaki2013}, and subsequent work addressed the consensus problem in these models, e.g., \cite{mobile_Banu2012,mobile_bonomi,mobile_bonomi2}. These models are all defined over the \textit{round-based computation} system (to be formally defined in Section \ref{s:model}), and they differ in two main dimensions \cite{mobile_bonomi,mobile_bonomi2}: (i) at which point in a round, faults can ``move'' to other nodes? and (ii) does a node have a knowledge when it is cured (i.e., after a fault moves to another node)? In this paper, we adopt the model proposed by Garay \cite{mobile_garay}:


\begin{itemize}
	\item At the beginning of round $t$, the Byzantine adversary picks the set of up to $f$ nodes that behave faulty in round $t$, and
	\item Once a node is \textit{cured} (i.e., the node that was faulty in the previous round, and becomes fault-free in the current round), it is aware of the condition and can remain silent to avoid disseminating faulty information.
\end{itemize}

Recently, in Garay's model, Banu et al. \cite{mobile_Banu2012} proposed an exact Byzantine consensus algorithm for at least $4f+1$ nodes, and Bonomi et al. \cite{mobile_bonomi,mobile_bonomi2} proposed an iterative approximate Byzantine consensus algorithm for at least $4f+1$ nodes. Bonomi et al. also proved that for a constrained class of \textit{memory-less} algorithms, their iterative algorithm is optimal. In this paper, we show that $4f+1$ is \underline{not tight} for a more general class of algorithms. In particular, we present an approximate consensus algorithm that requires only $\lceil 7f/2\rceil + 1$ nodes.

\subsubsection{Mobile Faults and Round-based Algorithms}

The mobile Byzantine fault model considered in this paper is defined over round-based algorithms, in which the system proceeds in synchronous rounds that consist of three steps: \textit{send, receive, compute} \cite{mobile_bonomi,mobile_bonomi2,mobile_garay}. There are three types of nodes in the system: \textit{faulty, healthy,} and \textit{cured}. For a slight abuse of terminology, we also call \textit{healthy} and \textit{cured} nodes as\textit{ fault-free} nodes. 
In the round-based algorithms, each {\em fault-free} node maintains a special state variable $v$. After a sufficient number of rounds, the state variable $v$ can be viewed as the \textit{output} of the fault-free nodes.\footnote{Using the technique from \cite{abraham_04_3t+1_async}, nodes can also estimate the number of required rounds and decide when to ``output'' the state variable $v$.} With mobile faults, each node may become Byzantine faulty and have its local storage (including the state variable and other bookkeeping variables) corrupted in any round. When a node is cured, it needs to recover its state variable and potentially other information. Therefore,  for a given round, we are only interested in the state variable $v$ at the \textit{healthy} nodes, since if majority of nodes remain healthy, \textit{cured} nodes can easily learn a fault-free state variable from other nodes. 

\subsubsection{Approximate Consensus}

Approximate consensus can be related to many distributed computations in emerging areas, such as data aggregation \cite{Kempe_gossip}, decentralized estimation  \cite{noisy_link}, and flocking \cite{jadbabaie_concensus}; hence, the problem of reaching approximate consensus in the presence of Byzantine faults has been studied extensively, including synchronous systems \cite{AA_Dolev_1986}, asynchronous systems \cite{abraham_04_3t+1_async}, arbitrary networks \cite{vaidya_PODC12}, transient link faults \cite{Tseng_netys14}, and time-varying networks \cite{aa_acc} $\cdots$ etc. Bonomi et al. \cite{mobile_bonomi,mobile_bonomi2} are among the first to study approximate consensus algorithms in the presence of mobile Byzantine fault models.

Roughly speaking, the round-based approximate consensus algorithms of interest have the properties below, which we will define formally in Section \ref{s:model}:

\begin{itemize}
	\item \textit{Initial state} of each node is equal to a real-valued input provided to that node.
	\item {\em Validity}: after each round of an algorithm, the state variable $v$ of each healthy node must remain in the range of the initial values of fault-free nodes.
	\item {\em Convergence}: for $\epsilon > 0$, after a sufficiently large number of rounds, the state variable of the healthy nodes are guaranteed to be within $\epsilon$ of each other.
\end{itemize}

\subsubsection{Main Contribution}

\begin{itemize}
	\item  We propose an approximate consensus algorithm that requires only $\lceil 7f/2\rceil + 1$ nodes. The algorithm relies on ``confession'' (a mechanism to ask others to ignore past behavior) and a variant of reliable broadcast (to learn  information from other \textit{healthy} nodes reliably). The technique may be applied to other problems under the mobile fault models.
	
	\item We show that the proposed algorithm is \textit{optimal} within a family of round-based algorithms, i.e., which only allows nodes to ``remember'' what happened in the previous rounds (but not the entire execution history).
\end{itemize}

\section{Related Work}

There is a rich literature on consensus-related problems \cite{AA_nancy,welch_book}. Here, we only discuss two most relevant categories.

\subsubsection{Exact Consensus under Mobile Byzantine Faults}

References \cite{mobile_Bonnet2014,mobile_buhrman,mobile_garay,mobile_Sasaki2013,mobile_Banu2012} studied the problem of reaching {\em exact} consensus under different mobile Byzantine fault models. In exact consensus algorithms, every fault-free node reaches exactly the \textit{same} output. Garay is among the first to study mobile faults \cite{mobile_garay}. In his model, the faults can ``move'' freely, and the cured nodes are aware of its condition. Garay proposed an algorithm requiring $6f+1$ nodes \cite{mobile_garay}. Later, Banu et al. \cite{mobile_Banu2012} improved the fault-tolerance level to $4f+1$ nodes. References \cite{mobile_Bonnet2014,mobile_Sasaki2013} considered a mobile fault model in which nodes are \textit{not} aware when they are cured. Sasaki et al. \cite{mobile_Sasaki2013} presented an algorithm requiring at least $6f+1$ nodes, whereas, Bonnet et al. \cite{mobile_Bonnet2014} proposed an algorithm requiring at least $5f+1$ nodes, and proved that $5f+1$ is tight in their fault model. Reference \cite{mobile_buhrman} also assumed that the nodes has the knowledge when it is cured; however, the ability of the adversary is more \textit{constrained} than the above models. The adversary \textit{cannot} choose an arbitrary set of nodes to be faulty, i.e., the faults can only ``move'' with message dissemination. Buhrman et al. \cite{mobile_buhrman} presented an optimal algorithm that requires $3f+1$ nodes. References \cite{mobile_Bonnet2014,mobile_buhrman,mobile_garay,mobile_Sasaki2013,mobile_Banu2012} considered only exact consensus; hence, the techniques are very different from the one in this paper. Moreover, to the best of our knowledge, we are the first to show that (approximate) consensus is solvable with only $\lceil 7f/2 \rceil+1$ nodes in Garay's model.



\subsubsection{Approximate Consensus}

Approximate consensus can be related to many distributed computations in networked systems, e.g., \cite{Kempe_gossip,noisy_link,jadbabaie_concensus}. Since many networked systems are tend to be fragile, the problem of reaching approximate consensus in the presence of Byzantine faults has been studied extensively. Most work assumed the ``fixed'' fault model; that is, once the Byzantine adversary picks a faulty node, then throughout the execution of the algorithm, the node remains faulty and will not be cured. Dolev et al. studied the problem in both synchronous and asynchronous systems \cite{AA_Dolev_1986}. Dolev et al. proposed an optimal synchronous algorithm, but the asynchronous one requires at least $5f+1$ nodes, which is only optimal within the family of \textit{iterative} algorithms.
Later, Abraham et al. proposed an optimal asynchronous algorithm that requires only $3f+1$ nodes \cite{abraham_04_3t+1_async}, which is optimal for all \textit{general} algorithms. The technique in this paper is inspired by the usage of ``witness'' and reliable broadcast in \cite{abraham_04_3t+1_async}; however, due to different synchrony assumptions and fault models, our technique differs from the ones in \cite{abraham_04_3t+1_async} (we will address more details in Sections \ref{s:alg} and \ref{s:key_lemma}). 

Kieckhafer and Azadmanesh studied the behavior of iterative algorithms (i.e., \textit{memory-less} algorithms) and proved some lower bounds under Mixed-Mode faults model, where nodes may suffer crash, omission, symmetric, and/or asymmetric Byzantine failures \cite{mixed_Kieckhafer}. Researchers also studied iterative approximate consensus under different communication assumptions, including arbitrary communication networks \cite{vaidya_PODC12}, networks with transient link faults \cite{Tseng_netys14}, and time-varying networks \cite{aa_acc} $\cdots$ etc. These works only assumed fixed fault model.

Bonomi et al. \cite{mobile_bonomi,mobile_bonomi2} are among the first to study approximate consensus algorithms in the presence of mobile Byzantine fault models. They presented optimal \textit{iterative} algorithms under different mobile fault models, and they proposed a mapping (or reduction) from the existing mobile Byzantine models to the Mixed-Mode faults model \cite{mixed_Kieckhafer}. As we will show later in this paper, the bound does not hold for a more general class of algorithms. In other words, the ``memory'' from previous rounds helps improve the fault-tolerance level. This paper essentially demonstrates how to use the ``memory'' effectively.

\section{Preliminary}

\subsection{Models and Round-based Algorithms}
\label{s:model}

\subsubsection{System Model}

We consider a synchronous message-passing system of $n$ nodes. The communication is through a point-to-point network, in which each pair of nodes is connected by a direct communication link. All the links are assumed to be reliable, and the messages \textit{cannot} be forged by the adversary. We assume that $n \geq \lceil 7f/2 \rceil + 1$, where $f$ is the upper bound on the number of faulty nodes in a given round. 

\subsubsection{Round-based Algorithms}

As in the prior work \cite{mobile_garay,mobile_Banu2012,mobile_bonomi,mobile_bonomi2}, we consider the round-based algorithms in this paper. The algorithm consists of three steps:

\begin{itemize}
	\item {\em Send}: send \textit{one} message to all other nodes
	\item {\em Receive}: receive the messages from other nodes
	\item {\em Compute}: based on the messages and local states, perform local computation
\end{itemize}

In addition, each node also maintains a special state variable $v$ such that after a sufficient number of rounds, the state becomes the \textit{output} at the node. Note that Bonomi et al. only considered iterative algorithms \cite{mobile_bonomi,mobile_bonomi2}, in which each node only sends and keeps a real-value state at all time, and there is no other information maintained (i.e., \textit{memory-less} algorithms or \textit{iterative} algorithms), whereas, we and references \cite{mobile_garay,mobile_Banu2012} consider a more general types of algorithms, where nodes may send and keep arbitrary state information.

\subsubsection{Mobile Byzantine Fault Model}

In this paper, we consider mobile Byzantine fault model proposed by Garay \cite{mobile_garay}. There are three types of nodes:

\begin{itemize}
	\item {\em Byzantine nodes}: in the beginning of each round, up to $f$ nodes may become Byzantine faulty. A Byzantine faulty node may misbehave arbitrarily, and the local storage may be corrupted. Possible misbehavior includes sending incorrect and mismatching (or inconsistent) messages to different nodes. We consider an omniscient adversary -- a single adversary that control which set of nodes would become faulty. Moreover, the Byzantine adversary is assumed to have a complete knowledge of the execution of the algorithm, including the states of all the nodes, contents of messages the other nodes send to each other, and the algorithm specification.
	
	~
	
	\item {\em Cured nodes}: a node is ``cured'' in the current round if it was faulty in the previous round, and becomes fault-free in the beginning of the current round. Under the model, a cured node has the knowledge that it just got cured, and hence can choose to stay silent at the current round, since the local states are potentially corrupted. A cured node follows the algorithm specification -- it receives messages and performs local computation accordingly.
	
	~
	
	\item {\em Healthy nodes}: all the other nodes belong to the set of healthy nodes. Particularly, they follow the algorithm specification, and the local storage is \underline{not} corrupted in the past and current rounds.
\end{itemize}

\subsection{Notation}

\paragraph{Nodes:}
To facilitate the discussion, we introduce the following notations to represent sets of nodes throughout the paper:

\begin{itemize}
	\item $\faulty[t]$: the set of nodes that are faulty in round $t$
	\item $\cured[t]$: the set of nodes that are cured in round $t$
	\item $\healthy[t]$: the set of nodes that are healthy in round $t$		
\end{itemize}
Nodes in $\healthy[t] \cup \cured[t]$ are said to be {\em fault-free} in round $t$.

\paragraph{Values:}
Given a given round $t$, let us define $v[t], \M[t]$ and $\m[t]$:

\begin{itemize}
	\item $v_i[t]$ is the special state variable (that later will be the output) maintained at node $i$ \underline{in the end of round $t$}. Notation $v_i[0]$ is assumed to be the input given to node $i$. For brevity, when the round index or node index is obvious from the context, we will often ignore $t$ or $i$.
	
	~
	
	\item $\M[t] = \max_{i \in \healthy[t] \cup \cured[t]} v_i[t]$. Notation $\M[t]$ is the largest state variable among the fault-free nodes at the end of round $t$. Since the initial state of each node is equal to its input, $\M[t]$ is equal to the maximum value of the initial input at the fault-free nodes.
	
	~
	
	\item $\m[t] = \min_{i \in \healthy[t] \cup \cured[t]} v_i[t]$. Notation $\m[t]$ is the smallest state variable among the fault-free nodes at the end of round $t$. Since the initial state of each node is equal to its input, $\m[t]$ is equal to the minimum value of the initial input at the fault-free nodes.	
\end{itemize}

\subsection{Correctness of Round-based Approximate Algorithms}

We are now ready to formally state the \textit{correctness condition} of round-based approximate algorithms under the mobile Byzantine fault model:

\begin{itemize}
	\item {\em Validity}: $\forall t > 0,$ 
	
	\[
	\m[t] \geq \m[0]~~~\text{and}~~~\M[t] \leq \M[0]
	\]

	~
	
	\item {\em Convergence}: for a given constant $\epsilon$, there exists an $t$ such that 
	\[
	\M[t] - \m[t] < \epsilon
	\]
\end{itemize}

\section{Algorithm CC}
\label{s:alg}

We now present \textit{Algorithm CC (Consensus using Confession)}, a round-based approximate algorithm. Throughout the execution of the algorithm, each node $i$ maintains a special state variable $v_i$. Recall that $v_i[t]$ represents the state at node $i$ \underline{in the end of round $t$} (i.e., after the state variable is updated). The convergence condition requires the state variables $v_i[t]$ to converge for a large enough $t$. 

Similar to the algorithms in \cite{mobile_garay,abraham_04_3t+1_async}, Algorithm CC proceeds in \textit{phases}. There are two phases in the algorithm: in the first phase (\textit{Collection Phase}), nodes exchange their state variables $v$ and construct a vector $E$ that stores others' state variables. $E_i[j]$ represents the value that $i$ receives from $j$. If $j$ is faulty or cured, $E_i[j]$ may not be the state variable at node $j$. The second phase (\textit{Confession Phase}) has three functionalities:

\begin{itemize}
	\item Exchange the vector $E$ constructed in the \textit{Collection Phase}. If a node $i$ is cured in the beginning of this phase (round $t+1$), then it sends $\emptyset$ to ``confess'' to all other fault-free nodes that it was faulty and subsequently, fault-free nodes will ignore messages from node $i$ from the previous round. If node $i$ is faulty, it may choose to send confession to only a subset of nodes; however, as long as we have enough redundancy, such misbehavior can be tolerated.
	
	~
	
	\item Construct a vector $V$ of ``trustworthy'' state variables. $V_i[j]$ represents the value what $i$ believe is $v_j[t-1]$, the state variable at node $j$ in the end of round $t-1$. A value $u$ from node $j$ is ``trustworthy'' if node $j$ does not confess (\textit{Condition 2} below), and enough nodes confess or ``endorse'' the value $u$  (\textit{Condition 1} below). Node $k$ is said to endorse the value $u$ if node $k$ does not confess, sends legit message, and has $E_k[j] = u$. Node $k$ may or may not be healthy.
	
	~
	
	\item Update the local state variable using the $\reduce$ function on the vector $V$. The $\reduce$ function is designed to trim enough values from $V$ so that none of the extreme values proposed by faulty nodes is used.
\end{itemize}

\subsection{Algorithm Specification}

\dividerline
\textbf{Algorithm CC: Steps to be executed by node $i$ in round $t$ for $t \geq 0$}
\dividerline
\begin{itemize}
	\item {\bf Round $t$}: \hfill {\em +++ Collection Phase +++}
	\begin{itemize}
		\item {\em Send}:
		
		~~~if \textit{$i$ is cured},
		
		~~~~~~~~send $(\perp, i)$
		
		~~~otherwise, send $(v_i[t-1],i)$\vspace*{5pt}
		
		\item {\em Receive}:\footnote{If nothing is received from $j$, then $u$ is assumed to be $\perp$, a null value. Also, we assume that a node can send a message to itself.} 
		
		~~~receive $(u,j)$ from node $j$ \vspace*{5pt}
		
		\item {\em Compute}: 
		
		\begin{itemize}
			\item $E_i[j] \gets u$
			
			\item if \textit{$i$ is healthy},
			
			~~~~$v_i[t] \gets v_i[t-1]$
			
		\end{itemize}
	\end{itemize}
	
	~
	
	\item {\bf Round $t+1$}: \hfill {\em +++ Confession Phase +++}
	\begin{itemize}
		\item {\em Send}: 
		
		~~~if \textit{$i$ is cured}, 
		
		~~~~~~~~send $(\emptyset,i)$\hfill \textit{$\slash \slash$Comment: ``confess'' faulty behavior}
		
		~~~otherwise, send $(E_i,i)$ \vspace*{5pt}
		
		\item {\em Receive}:
		
		~~~if a legitimate tuple $(E_j,j)$ is received from node $j$,\footnote{Here, $E_j = \emptyset$ is legitimate.}
		
		\[
		R_i[j] \gets E_j
		\]

		\item {\em Compute}: 
		
		\begin{itemize}
			\item if the following two conditions are satisfied:
			
			\begin{itemize}
				\item {\em Condition 1}: ~~~~~there are $\geq n-f$ distinct nodes $k$ such that (i) $R_i[k] = E_k \neq \emptyset$ and $E_k[j] = u$, or (ii) $R_i[k] = \emptyset$ \vspace*{5pt}
				\item {\em Condition 2}: ~~~~~$R_i[j] \neq \emptyset$
			\end{itemize}
			
			~
			
			then \hfill {\em $\slash \slash$Comment: $u$ is ``trustworthy''}
			
			\[
			V_i[j] \gets u
			\]
			
			otherwise,
			
			\[
			V_i[j] \gets \perp
			\]
			
			\item update state variable as follows:
			
			\[
			v_i[t+1] \gets \reduce(V_i)
			\]
		\end{itemize}
	\end{itemize}		
\end{itemize}

\dividerline

\subsection{Reduce Function}

Reduce function is widely used in iterative approximate Byzantine consensus algorithms, e.g., \cite{abraham_04_3t+1_async,mobile_bonomi,mobile_bonomi2,mixed_Kieckhafer,AA_Dolev_1986}. We adopt the same structure: order the values, trim potentially faulty values, and update local state.
Different from the prior work, our $\reduce$ function trims different number of values at each round. The exact number depends on the number of $\perp$ values received. A $\perp$ value may be a result of faulty behavior or a confession. Below, we define the number of values to be trimmed.

\begin{definition}
	\label{def:num}
	Suppose that node $i$ receives $x$ $\perp$ values in the vector $V_i$ at round $t+1$. Then, define
	
	\[
	\num_i = 
	\begin{cases}
	f, & \text{if}\ x \leq f \\
	\lceil f-\frac{x-f}{2} \rceil, & \text{otherwise}\
	\end{cases}
	\]
\end{definition}

The value $\num$ counts the number of potentially faulty values in the vector $V_i$. In general, the more confessions that $i$ sees in $V_i$, the less faulty values are in $V_i$. Lemma \ref{lemma:num} formally shows that $\num_i$ is large enough to trim all the extreme values proposed by faulty nodes. Now, we present our $\reduce$ function below: 


\dividerline
\textbf{Reduce function: $\reduce(V_i)$ at node $i$}
\dividerline
\begin{itemize}
	\item Calculate $\num_i$ as per Definition \ref{def:num}.\vspace*{3pt}
	\item Remove all $\perp$ values in $V_i$. Denote the new vector by $V_i'$.\vspace*{3pt}
	\item Order $V_i'$ in a non-decreasing order. Denote the ordered vector by $O_i$.\vspace*{3pt}
	\item Trim the bottom $\num_i$ and the top $\num_i$ values in $O_i$. In other words, generate a new vector containing the values $O_i[\num_i+1], O_i[\num_i+2], \cdots, O_i[|O_i|-\num_i-1]$. Denote the trimmed vector by $O_i^t$.\vspace*{3pt}
	\item Return
	
	\begin{equation}
	\label{eq:reduce}
	\frac{\min(O_i^t) + \max(O_i^t)}{2}
	\end{equation}
\end{itemize}
\dividerline

\section{Analysis}

\subsection{Key Properties of $V$}
\label{s:key_lemma}

Before the $\reduce$ function is executed, the vector $V$ at all fault-free nodes satisfies nice properties as stated in the lemmas below. The first four lemmas (Integrity I-IV) show that Algorithm CC achieves properties similar to reliable broadcast \cite{abraham_04_3t+1_async} -- all fault-free nodes are able to see identical values in $V$ if the sender node is either healthy or cured. Reliable broadcast in \cite{abraham_04_3t+1_async} also guarantees \textit{Uniqueness} -- if the value sent from a node is not $\perp$, then the value appears identically in all fault-free node's $V$ vector. However, the $V$ vectors in Algorithm CC may still contain faulty values, since a faulty node that just moved in round $t+1$ can send different $E$ vectors to different fault-free nodes to ``endorse'' different values. This is the main reason that why Algorithm CC requires more than $3f+1$ nodes. In the proofs below, we will often denote $v_i[t-1]$ by $v$ for brevity. The indices should be clear from the context.

\begin{lemma} {\bf (Integrity I)}
	\label{lemma:integrityI}
	If node $i$ is healthy in both rounds $t$ and $t+1$, then for all fault-free $j \in \healthy[t+1] \cup \cured[t+1]$, $V_j[i] = v_i[t-1]$, the value sent by node $i$ in round $t$.
\end{lemma}

\begin{proof}
	
	Fix a node $i \in \healthy[t] \cap \healthy[t+1]$ which sends the value $v_i[t-1]$ in round $t$.
	In the receive step of round $t$, each node $k \in \healthy[t] \cup \cured[t]$ receives the value and has $E_k[i] = v$. By definition, $|\healthy[t] \cup \cured[t]| \geq n-f$. 
	Suppose in the beginning of round $t+1$, $b \leq f$ of the mobile Byzantine faults move to the nodes in $\healthy[t] \cup \cured[t]$.
	Then, observe that 
	
	\begin{itemize}
		\item $|\healthy[t] \cup \cured[t]|-b$ healthy nodes send a legitimate tuple to all other nodes, and  $E_k \neq \emptyset$ and $E_k[i] = v$ for node $k \in \healthy[t]\cup\cured[t] - \faulty[t+1]$. Denote this set of healthy nodes by $A$.
		
		\item Since $b$ mobile faults move in round $t+1$, exactly $b$ nodes are cured and send the confession ($\emptyset$) in round $t+1$. Denote this set of cured nodes by $B$.
	\end{itemize}
	Note that nodes in $A \cup B$ are either cured or healthy; hence, all fault-free nodes will observe their behavior identically.
	
	Now, consider a node $j \in \healthy[t+1] \cup \cured[t+1]$. From its perspective, {\em Condition 1} in the compute step in round $t+1$ is met due to the observations above and the fact that $|A|+|B| \geq (|\healthy[t] \cup \cured[t]|-b) + b = |\healthy[t] \cup \cured[t]| \geq n-f$. Moreover, by definition, $i$ is healthy in round $t+1$; hence, {\em Condition 2} is also met. Therefore, node $j$ will have $V_j[i] = v$. \fillbox
\end{proof}

\begin{lemma} {\bf (Integrity II)}
	\label{lemma:integrityII}
	If node $i$ is healthy in round $t$ and becomes faulty in round $t+1$, then for all fault-free $j \in \healthy[t+1] \cup \cured[t+1]$, either $V_j[i] = \perp$ or $V_j[i] = v_i[t-1]$, the value sent by node $i$ in round $t$.
\end{lemma}

\begin{proof}
	The proof is by contradiction. Suppose that at some node $j \in \healthy[t+1] \cup \cured[t+1]$, $V_j[i] = u$ such that $u \neq \perp$. 
	Now, observe that:
	\begin{itemize}
		\item {\em Obs 1}: $V_j[i] = u$ only if there are enough node $k$ that endorses or confesses (Condition 1 in Algorithm CC). Denote this set of nodes by $W_u$. And we have $|W_u| \geq n-f$.
		\item {\em Obs 2}: Since node $i$ is healthy in round $t$, every node $k \in \healthy[t+1]$ did \textit{not} endorse value $u$ (they heard value $v$ and endorses $v$ in round $t$). 
		\item {\em Obs 3}: Obs 2 together with the fact that $|\healthy[t+1]| \geq n-2f$ imply that there are $\leq n-(n-2f) = 2f$ nodes in the set $W_u$.
	\end{itemize}
	We have $|W_u| \leq 2f < \lceil 7f/2 \rceil +1-f = \lceil 5f/2 \rceil +1$, contradicting Obs 1. \fillbox
	
\end{proof}	

\begin{lemma} {\bf (Integrity III)}
	\label{lemma:integrityIII}
	If node $i$ is cured in round $t$, then for all fault-free $j \in \healthy[t+1] \cup \cured[t+1]$, $V_j[i] = \perp$.
\end{lemma}

The proof is similar to the proof of Lemma \ref{lemma:integrityI} and omitted here for brevity.

\begin{lemma} {\bf (Integrity IV)}
	\label{lemma:integrityIV}
	If node $i$ is cured in round $t+1$, then for all fault-free $j \in \healthy[t+1] \cup \cured[t+1]$, $V_j[i] = \perp$.
\end{lemma}

\begin{proof}
	Since node $i$ is cured in round $t+1$, it will send the confession $(\emptyset)$ to all fault-free nodes in round $t+1$. Thus, for all $j \in \healthy[t+1] \cup \cured[t+1]$, $R_j[i] = \emptyset$, violating Condition 2. Therefore, $V_j[i] = \perp$. \fillbox
\end{proof}

~

The only case left is analyzing the behavior of nodes which remain faulty in both rounds $t$ and $t+1$. These nodes are indeed able to produce different values in the $V$ vectors at fault-free nodes; however, by construction, these nodes are limited in number. To see this, consider the following two scenarios:
\begin{itemize}
	\item When no faulty node moves, i.e., $\faulty[t] = \faulty[t+1]$. Then, all fault-free nodes receive identical $V$ vectors, since Condition 1 cannot be satisfied if faulty nodes send different values to different nodes in round $t$.
	\item When all faulty nodes move, i.e., $\faulty[t] \cap \faulty[t+1] = \emptyset$. Then, by Lemmas \ref{lemma:integrityII}, \ref{lemma:integrityIII}, and \ref{lemma:integrityIV}, no fault-free nodes will see different values in round $t+1$.
\end{itemize}
The lemma below characterizes the bound on the number of different values.

\begin{lemma}
	\label{lemma:limit}
	Suppose $n \geq \lceil 7f/2 \rceil+1$. For a pair of fault-free nodes $i, j \in \healthy[t+1] \cup \cured[t+1]$, at most $\lfloor f/2 \rfloor-1$ non-$\perp$ values differs in $V_i$ and $V_j$. In other words, there are $\geq n - \lceil f/2 \rceil +1$ identical values in $V_i$ and $V_j$.
\end{lemma}

\begin{proof}
	The proof is by contradiction. Suppose that there exists a pair of fault-free nodes $i, j$ such that $\lfloor f/2 \rfloor$ different values appear in $V_i$ and $V_j$. Consider the value belonging to some node $k$, i.e., $V_i[k] \neq V_j[k]$ and $V_i[k], V_j[k] \neq \perp$. Then, we can make the following observations:
	
	\begin{itemize}
		\item {\em Obs 1}: By Lemmas \ref{lemma:integrityI}, \ref{lemma:integrityII}, \ref{lemma:integrityIII}, and \ref{lemma:integrityIV}, node $k$ must remain faulty in both rounds $t$ and $t+1$.
		
		\item {\em Obs 2}: By Condition 1, there are $\geq n-f$ nodes that send the value $V_i[k]$ or send the confession ($\emptyset$) to node $i$ in round $t+1$. Denote this set of nodes by $W_i$. For easiness of discussion, let us call these nodes the ``witnesses'' of the value $V_i[k]$.
		
		\item {\em Obs 3}: By assumption and Obs 1, at most $\lceil f/2 \rceil$ faults move from round $t$ to round $t+1$.
		
		\item {\em Obs 4}: Among the nodes in $W_i$, at least $|W_i| - (f + \lceil f/2 \rceil)$ are nodes that are healthy in both rounds $t$ and $t+1$. This is because (i) by Obs 3, at most $\lceil f/2 \rceil$ faults move, and (ii) cured node $l$ in round $t+1$ (i.e., $l \in \cured[t+1]$) send the confession ($\emptyset$) in round $t+1$, which result into $R_i[l] = R_j[l] = \emptyset$ in the receive step of round $t+1$.
	\end{itemize}
	
	Now, consider node $j$. By Obs 4, it has $\leq n-(|W_i| - (f + \lceil f/2 \rceil))$ witnesses of the value $V_j[k]$, since this is the number of nodes that are healthy in both rounds $t$ and $t+1$ and endorses the value $V_i[k]$. Denote this set of witness of the value $V_j[k]$ by $W_j$. Then, we have
	
	\begin{align*}
	|W_j| &\leq n-(|W_i| - (f +  \lceil f/2 \rceil)) \\
	&\leq n-((n-f)-(f +  \lceil f/2 \rceil)) = \lceil 5f/2 \rceil ~~~~~~\text{by Obs 2} \\
	&< \lceil 5f/2 \rceil+1 = n-f
	\end{align*}
	
	Therefore, Condition 1 is not satisfied at node $j$; hence, $V_j[k]$ can only be either the value $V_i[k]$ or $\perp$, a contradiction.	\fillbox
\end{proof}	

\commentOut{++++++
	The next lemma presents a useful observation on faulty values -- the non-$\perp$ value sent by faulty nodes.
	
	\begin{lemma}
		\label{lemma:f-b}
		If $b \leq f$ faults move in round $t+1$, then there are $\leq f-b$ faulty values in $V_i$ in round $t+1$ for each fault-free node $i \in \healthy[t+1] \cup \cured[t+1]$.
	\end{lemma}
	
	\begin{proof}
		This lemma follows from Lemma \ref{lemma:integrityII} and the fact that there are at most $f$ faulty nodes.
	\end{proof}
	+++++}

\subsection{Correctness}

For brevity, we only prove the correctness properties for healthy nodes, since cured nodes will have valid state variables if they remain fault-free in the next round. We begin with a useful lemma on $\num$ (as per Definition \ref{def:num}). Here, a faulty value is the non-$\perp$ value sent by faulty nodes.

\begin{lemma}
	\label{lemma:num}
	For a given odd round $t \geq 1$ and $i \in \healthy[i]$, there are at most $\num_i$ faulty values in $V_i[t]$.
\end{lemma}

\begin{proof}
	If $V_i[t]$ contains $\leq f \perp$ values, then the lemma holds by assumption. Now, consider the case when there are $x \perp$ values in $V_i[t]$, where $x > f$. There are only three ways to produce $\perp$ values: (i) by cured nodes in round $t-1$ (due to Lemma \ref{lemma:integrityIII}), (ii) by cured nodes in round $t$ (due to due to Lemma \ref{lemma:integrityIV}), and (iii) by faulty nodes in round $t$.
	Assume that $b$ faults move in round $t$ and $b'$ faulty nodes produce $\perp$ values. Observe that (i) at most $f$ cured nodes in round $t-1$, (ii) exactly $b$ cured nodes in round $t$, and (iii) exactly $f-(b+b')$ faulty values in $V_i[t]$. Then, we have
	\begin{align*}
	\num_i 	&= \lceil f-\frac{x-f}{2} \rceil~~~~~~~~~~~~~~~~~~~~~~~~~~~~\text{by Definition \ref{def:num}}\\
	&\geq \lceil f-\frac{(b+f+b')-f}{2} \rceil = \lceil f- \frac{b+b'}{2}\rceil~~~~~~\text{by observations above}\\
	&\geq f-(b+b') = \text{number of faulty values in}\ V_i~~~~~~~~~~~~~~~~~~~~~~~~\text{\fillbox}
	\end{align*}	
	
\end{proof}

\begin{lemma}
	{\bf (Validity)}
	\label{lemma:validity}
	For a given round $t \geq 0$, if $i \in \healthy[t]$, then
	
	\[
	\M[0] \geq v_i[t] \geq \m[0]
	\]
\end{lemma}	

\begin{proof}
	The proof is by induction on the number of rounds. 
	\begin{itemize}
		\item {\em Initial Step}: When $t = 0$, the statement holds, since by definition, \\$~~~\M[0] \geq v_i[0] \geq \m[0]$.
		
		\item {\em Induction Step}: suppose the statement holds for some $h > 0$, consider round $h+1$. If $h$ is a \textit{Collection Round} ($h$ is even), then the statement holds trivially, since $v_i[h] \gets v_i[h-1]$ in the compute step.
		Now, consider the case when $h$ is odd ($h$ is an \textit{Update Round}). 
		Lemma \ref{lemma:num} implies that in the trim step of the $\reduce$ function (the fourth step), all the faulty values will be trimmed if they are too large or too small. Therefore, the maximal and minimal values in $O_i^t$ will always be inside the range of the maximal and minimal values of state values of the fault-free nodes in round $h$. Hence the return value of the $\reduce$ function satisfies Validity by the induction hypothesis.	\fillbox
	\end{itemize}
	
\end{proof}	

Before proving convergence, we show a lemma that bounds the range of the updated state variables. Recall that $\M[t]$ and $\m[t]$ represent the maximal and minimal state variables, respectively, at healthy nodes in round $t$. We only care about the state variables in the even round, since in the odd round, the state variable remains the same at healthy nodes.

\begin{lemma}
	\label{lemma:shrink}
	For some even integer $t > 0$, we have
	\[
	\M[t+1] - \m[t+1] \leq \frac{\M[t-1] - \m[t-1]}{2}
	\]
\end{lemma}	

\begin{proof}
	To prove the lemma, we need to show that for any pair of fault-free nodes $i, j$, we have
	
	\begin{equation}
	\label{eq:ij}
	| v_j[t+1] - v_i[t+1] | \leq  \frac{\M[t-1] - \m[t-1]}{2} 
	\end{equation}
	
	Let $V_i$ and $V_j$ denote the $V$ vectors at $i$ and $j$, respectively, at the compute step (the third step) of round $t+1$.
	Then, define $R = V_i \cap V_j$. Recall that $O^t_i$ and $O^t_j$ represent the trimmed vector in the $\reduce$ function at $i$ and $j$, respectively. Then, we have the following key claim:
	
	\begin{claim}
		\label{claim:media}
		Let $m$ be the median of the values in $R$. Then, $m \in O^t_i$ and $m \in O^t_j$.
	\end{claim}	
	
	\begin{proof}
		We make the following observations:
		\begin{itemize}
			\item {\em Obs 1}: By Lemma \ref{lemma:limit}, $|R| \geq n - \lceil f/2 \rceil + 1 \geq 3f+1$.
			\item {\em Obs 2}: Suppose there are $x \perp$ values in $R$. Consider two cases:
			\begin{itemize}
				\item Case I: if $x \leq f$, then $m \in O^t_i$, because after removing $f \perp$ values from $V_i$, we trim $f$ elements from each side. Similarly, we can show  $m \in O^t_j$.
				
				\item Case II: if $x \geq f$, then $m \in O^t_i$, because after removing $x \perp$ values from $V_i$, we trim $\num_i$ elements from each side. In other words, we trim at most
				
				\[
				x+2*\num_i = x+2\lceil (f-\frac{x-f}{2}) \rceil = x+2f-x+f=3f
				\]
				
				values from $R$. This together with Obs 1 implies that $m \in O^t_i$. Similarly, we can show  $m \in O^t_j$.
			\end{itemize}
			These two cases proves the claim. \fillbox
		\end{itemize}
	\end{proof}	
	
	The rest of the proof of Lemma \ref{lemma:shrink} follows from the claim using the standard tricks from prior work, e.g., \cite{AA_nancy,abraham_04_3t+1_async,vaidya_PODC12}. We include the proof in Appendix \ref{a:lemma:shrink}. \fillbox
\end{proof}	

Lemma \ref{lemma:shrink} and simple arithmetic operations imply the following:

\begin{lemma}
	{\bf (Convergence)}
	\label{lemma:converge}
	Given a $\epsilon > 0$, there exists a round $t$ such that $\M[t]-\m[t] < \epsilon$.
\end{lemma}	

Lemmas \ref{lemma:validity} and \ref{lemma:converge} imply that Algorithm CC is correct:

\begin{theorem}
	{\bf (Correctness)}
	Algorithm CC solves approximate consensus in under Garay's model given that $n \geq \lceil 7f/2 \rceil +1$.
\end{theorem}

\section{Impossibility Result}
\label{s:impossible}

This section proves that for a certain family of round-based algorithms, $\lceil 7f/2 \rceil +1$ is the lower bound on the number of nodes (fault-tolerance level), proving that Algorithm CC is optimal within this family of algorithms.

\subsubsection{$2$-Memory Round-based Algorithms}
\label{s:2-alg}

As discussed before, the iterative algorithms considered in \cite{mobile_bonomi,mobile_bonomi2,AA_Dolev_1986,vaidya_PODC12} are \textit{memory-less}, i.e., it can only send its own state, and it updates state in every round. As proved in \cite{mobile_bonomi,mobile_bonomi2}, such type of memory-less algorithms requires $4f+1$ nodes. For the lower bound proof, we consider a slightly more general type of algorithms -- \textit{$2$-memory round-based algorithms} -- in which nodes can send arbitrary messages, carry information from the previous round, but nodes have to update their state variables every two rounds (hence, the name $2$-memory). While the definition seems constrained, many Byzantine consensus algorithms belong to this family of algorithms, e.g., \cite{mobile_garay,mobile_Banu2012,abraham_04_3t+1_async}. Note that the original algorithm proposed by Lamport, Shostak, and Pease \cite{lamport_agreement2} \textit{does not} belong to $2$-memory round-based algorithms, as nodes collect many more rounds of information before updating their state variables.

\subsubsection{Lower Bound Proof}

The lower bound proof is similar to the lower bound proof for iterative algorithms, e.g., \cite{AA_Dolev_1986,vaidya_PODC12}; however, we also need to consider how faulty nodes move, which makes the proof slightly more complicated. Note that using Integrity I-IV (Lemmas \ref{lemma:integrityI}, \ref{lemma:integrityII}, \ref{lemma:integrityIII}, and \ref{lemma:integrityIV}), it is fairly easy to show that for $f=1$, Algorithm CC solves the problem for $n=3f+1=4$.

\begin{theorem}
	\label{thm:lower_bound}
	It is impossible for any $2$-memory round-based algorithm to solve approximate consensus under Garay's model if $n \leq \lceil 7f/2 \rceil$ and $f > 1$.
\end{theorem}

\begin{proof}
	Consider the case when $f=2$, and $n = 7$. Denote by the set of nodes $S = \{a, b, c, d, e, f, g\}$. For simplicity, assume that a node can be in the \textit{cured} phase in round $0$. Then, suppose in round $0$: $a, b$ are cured, $c, d$ are faulty, and $e, f, g$ are healthy. And, nodes $e, f$ has input $m$, and node $g$ has input $m'$, where $m' > m$ and $m' - m > \epsilon$.
	
	In round $0$, faulty nodes $c, d$ behave to nodes $a, e, f$ as if they have input $m$, and behave to nodes $b, g$ as if they have input $m'$. In the beginning of round $1$, the adversary moves the fault from node $d$ to node $e$; hence, $a, b, f, g$ are healthy, $c, e$ are faulty, $d$ is cured in round $1$. The new faulty node $e$ and the original faulty node $c$ behave in the following way (i) behave to nodes $a, d, f$ as if node $c, d, e$ have input $m$, (ii) behave to node $b, g$ as if nodes $c, d, e$ have input $m'$, and (iii) otherwise follow the algorithm specification.
	
	Now, from the perspective of node $f$, there are two scenarios:
	
	\begin{itemize}
		\item If nodes $c, d$ are faulty, then the fault-free inputs are $m, m, m'$, and
		\item If nodes $d, g$ are faulty, then the fault-free inputs are $m, m, m$, and
	\end{itemize}
	By assumption, node $e$ needs to update the state variable now and it could not distinguish from the two scenarios, since it cannot exchange more messages. Therefore, node $e$ must choose some value that satisfies the validity condition in \textit{both} scenarios, and the value is $m$.\footnote{There are other scenarios not discussed in the proof for brevity; however, $m$ is the only value works for each of the scenarios.} Therefore, in round $1$, the state variable at node $e$ remains $m$. We can show the same situation holds for node $a, d$.
	
	From the perspective of node $g$, there are also two scenarios:
	
	\begin{itemize}
		\item If nodes $c, d$ are faulty, then the fault-free inputs are $m, m, m'$, and
		\item If nodes $e, f$ are faulty, then the fault-free inputs are $m', m', m'$, and
	\end{itemize}	
	Then, node $g$ has to choose $m'$ to satisfy the validity condition in round $1$. Similarly, node $b$ has to choose $m'$.
	
	Then in round $2$, the adversary picks nodes $a, b$ to be faulty. Observe that this scenario is identical to round $0$: two cured nodes, two faulty nodes, and three healthy nodes with state variables $m, m,$ and $m'$. Therefore, the adversary can behave in the same way so that no healthy node will change their state variables; hence, convergence cannot be achieved.	\fillbox
\end{proof}

\section{Conclusion}

Under Garay's mobile Byzantine fault model \cite{mobile_garay}, we present an approximate consensus algorithm that requires only $\lceil 7f/2\rceil + 1$ nodes, an $\lfloor f/2 \rfloor$ improvement over the state-of-the-art algorithms \cite{mobile_bonomi,mobile_bonomi2}. Moreover, we also show that the proposed algorithm is \textit{optimal} within the family of $2$-memory round-based algorithms. Whether $\lceil 7f/2\rceil + 1$ is tight for general approximate algorithms remains open.


\centerline{\Large\bf Appendices}

\appendix

\section{Proof of Lemma \ref{lemma:shrink}}
\label{a:lemma:shrink}

The claim proved in Lemma \ref{lemma:shrink} implies that $\max(O^t_i) \geq m$, and by the last step of the $\reduce$ function and the fact that all the extreme value proposed by faulty nodes are trimmed (by Lemma \ref{lemma:num}), we have $\min(O^t_i) \geq \m[t-1]$. Therefore,

\[
v_i[t+1] \geq \frac{m + \m[t-1]}{2}
\]

Similarly, we can show $\min(O^t_i) \leq m$, $\max(O^t_i) \geq \M[t-1]$, and

\[
v_j[t+1] \leq \frac{m + \M[t-1]}{2}
\] 

Now, we need to show that (\ref{eq:ij}) holds. 	Without loss of generality, assume that $v_j[t+1] \geq v_i[t+1]$. Then, we have

\begin{align*}
v_j[t+1] - v_i[t+1] &\leq  \frac{m + \M[t-1]}{2} - \frac{m + \m[t-1]}{2}\\
&=\frac{\M[t-1] - \m[t-1]}{2} 
\end{align*}

This completes the proof.


\begin{thebibliography}{10}
	
	\bibitem{abraham_04_3t+1_async}
	I.~Abraham, Y.~Amit, and D.~Dolev.
	\newblock Optimal resilience asynchronous approximate agreement.
	\newblock In {\em OPODIS}, pages 229--239, 2004.
	
	\bibitem{welch_book}
	H.~Attiya and J.~Welch.
	\newblock {\em Distributed Computing: Fundamentals, Simulations, and Advanced
		Topics}.
	\newblock Wiley Series on Parallel and Distributed Computing, 2004.
	
	\bibitem{mobile_Banu2012}
	N.~Banu, S.~Souissi, T.~Izumi, A.~N. Bessani, M.~Correia, N.~F. Neves,
	H.~Buhrman, and J.~A. Garay.
	\newblock An improved byzantine agreement algorithm for synchronous systems
	with mobile faults.
	\newblock 2012.
	
	\bibitem{mobile_Bonnet2014}
	F.~Bonnet, X.~D{\'e}fago, T.~D. Nguyen, and M.~Potop-Butucaru.
	\newblock {\em Tight Bound on Mobile Byzantine Agreement}, pages 76--90.
	\newblock Springer Berlin Heidelberg, 2014.
	
	\bibitem{mobile_bonomi}
	S.~Bonomi, A.~D. Pozzo, M.~Potop{-}Butucaru, and S.~Tixeuil.
	\newblock Approximate agreement under mobile byzantine faults.
	\newblock {\em CoRR}, abs/1604.03871, 2016.
	
	\bibitem{mobile_bonomi2}
	S.~Bonomi, A.~D. Pozzo, M.~Potop-Butucaru, and S.~Tixeuil.
	\newblock Approximate agreement under mobile byzantine faults.
	\newblock In {\em 2016 IEEE 36th International Conference on Distributed
		Computing Systems (ICDCS)}, pages 727--728, June 2016.
	
	\bibitem{mobile_buhrman}
	H.~Buhrman, J.~A. Garay, and J.~H. Hoepman.
	\newblock Optimal resiliency against mobile faults.
	\newblock In {\em Twenty-Fifth International Symposium on Fault-Tolerant
		Computing. Digest of Papers}, pages 83--88, June 1995.
	
	\bibitem{AA_Dolev_1986}
	D.~Dolev, N.~A. Lynch, S.~S. Pinter, E.~W. Stark, and W.~E. Weihl.
	\newblock Reaching approximate agreement in the presence of faults.
	\newblock {\em J. ACM}, 33:499--516, May 1986.
	
	\bibitem{mobile_garay}
	J.~A. Garay.
	\newblock Reaching (and maintaining) agreement in the presence of mobile faults
	(extended abstract).
	\newblock In {\em Proceedings of the 8th International Workshop on Distributed
		Algorithms}, WDAG '94, pages 253--264, London, UK, UK, 1994. Springer-Verlag.
	
	\bibitem{aa_acc}
	A.~Haseltalab and M.~Akar.
	\newblock Approximate byzantine consensus in faulty asynchronous networks.
	\newblock In {\em 2015 American Control Conference (ACC)}, pages 1591--1596,
	July 2015.
	
	\bibitem{jadbabaie_concensus}
	A.~Jadbabaie, J.~Lin, and A.~Morse.
	\newblock Coordination of groups of mobile autonomous agents using nearest
	neighbor rules.
	\newblock {\em Automatic Control, IEEE Transactions on}, 48(6):988 -- 1001,
	june 2003.
	
	\bibitem{Kempe_gossip}
	D.~Kempe, A.~Dobra, and J.~Gehrke.
	\newblock Gossip-based computation of aggregate information.
	\newblock pages 482--491. IEEE Computer Society, 2003.
	
	\bibitem{mixed_Kieckhafer}
	R.~M. Kieckhafer and M.~H. Azadmanesh.
	\newblock Reaching approximate agreement with mixed-mode faults.
	\newblock {\em IEEE Transactions on Parallel and Distributed Systems},
	5(1):53--63, Jan 1994.
	
	\bibitem{lamport_agreement2}
	L.~Lamport, R.~Shostak, and M.~Pease.
	\newblock The byzantine generals problem.
	\newblock {\em ACM Trans. Program. Lang. Syst.}, 4(3):382--401, July 1982.
	
	\bibitem{AA_nancy}
	N.~A. Lynch.
	\newblock {\em Distributed Algorithms}.
	\newblock Morgan Kaufmann, 1996.
	
	\bibitem{mobile_Sasaki2013}
	T.~Sasaki, Y.~Yamauchi, S.~Kijima, and M.~Yamashita.
	\newblock {\em Mobile Byzantine Agreement on Arbitrary Network}, pages
	236--250.
	\newblock Springer International Publishing, Cham, 2013.
	
	\bibitem{noisy_link}
	I.~Schizas, A.~Ribeiro, and G.~Giannakis.
	\newblock Consensus in ad hoc {WSNs} with noisy links -- {Part I}: Distributed
	estimation of deterministic signals.
	\newblock {\em Signal Processing, IEEE Transactions on}, 56(1):350--364, Jan
	2008.
	
	\bibitem{Tseng_netys14}
	L.~Tseng and N.~H. Vaidya.
	\newblock Iterative approximate consensus in the presence of {Byzantine} link
	failures.
	\newblock In {\em Networked Systems - Second International Conference, {NETYS}
		2014, Marrakech, Morocco, May 15-17, 2014. Revised Selected Papers}, pages
	84--98, 2014.
	
	\bibitem{vaidya_PODC12}
	N.~H. Vaidya, L.~Tseng, and G.~Liang.
	\newblock Iterative approximate {Byzantine} consensus in arbitrary directed
	graphs.
	\newblock In {\em Proceedings of the thirty-first annual ACM symposium on
		Principles of distributed computing}, PODC '12. ACM, 2012.
	
	\bibitem{mobile_yung}
	M.~Yung.
	\newblock The mobile adversary paradigm in distributed computation and systems.
	\newblock In {\em Proceedings of the 2015 ACM Symposium on Principles of
		Distributed Computing}, pages 171--172. ACM, 2015.
	
\end{thebibliography}
\end{document}